\documentclass[11pt]{article}
\usepackage{amsmath}
\usepackage{amsfonts}
\usepackage{graphicx}
\usepackage{comment}
\usepackage{url}
\usepackage{float}
\usepackage{comment}
\usepackage{authblk}
\usepackage{algorithm}
\usepackage{algorithmic}
\usepackage{amsthm}
\usepackage{enumitem}
\usepackage{isomath}
\usepackage{caption}
\usepackage{subcaption}
\usepackage[margin=1in]{geometry}
\usepackage[margin=0.5in,font={small}]{caption}

\newtheorem{proposition}{Proposition}

\begin{document}

\title{Joint inference of genome structure and content in heterogeneous tumor samples}

\author{
	Andrew McPherson$^{1,2}$
	Andrew Roth$^{2,3}$,
	Gavin Ha$^{2,3}$,
	Sohrab P. Shah$^{2}$,
	Cedric Chauve$^{4}$,
	S.~Cenk~Sahinalp$^{1,5}$\\
	\footnotesize
	$^{1}$~School of Computing Science, Simon Fraser University, BC, Canada \\
	$^{2}$~Department of Molecular Oncology, BC Cancer Agency, BC, Canada \\
	$^{3}$~Bioinformatics Graduate Program, University of British Columbia, BC, Canada \\
	$^{4}$~Department of Mathematics, Simon Fraser University, BC, Canada \\
	$^{5}$~School of Informatics and Computing, Indiana University, Bloomington, IN, USA
}

\maketitle

\begin{abstract}
For a genomically unstable cancer, a single tumour biopsy will often contain a mixture of competing tumour clones.
These tumour clones frequently differ with respect to their genomic content (copy number of each gene) and structure (order of genes on each chromosome).
Modern bulk genome sequencing mixes the signals of tumour clones and contaminating normal cells, complicating inference of genomic content and structure.
We propose a method to unmix tumour and contaminating normal signals and jointly predict genomic structure and content of each tumour clone.
We use genome graphs to represent tumour clones, and model the likelihood of the observed reads given clones and mixing proportions.
Our use of haplotype blocks allows us to accurately measure allele specific read counts, and infer allele specific copy number for each clone.
The proposed method is a heuristic local search based on applying incremental, locally optimal modifications of the genome graphs.
Using simulated data, we show that our method predicts copy counts and gene adjacencies with reasonable accuracy.
\end{abstract}

\section{Introduction}

Human cells have evolved DNA repair mechanisms to mitigate the effects of DNA breakage during transcription and replication.
During the lifetime of many cancers, one or more of these mechanisms will be compromised.
With DNA repair compromised, DNA breakages will either go unrepaired or will be repaired by a less accurate but still functional mechanism.
Improper repair of DNA breakage events will result in structural chromosomal aberrations in descendent tumour cell lineages.
Structural aberrations can then lead to further problems during mitosis, with incorrect segregation of DNA to daughter cells resulting in numerical chromosomal aberrations \cite{Burrell:2013aa}.
If both daughter cells are viable, incorrect segregation leads to a divergence in chromosomal content between the two descendent lineages.
The progressive acquisition of structural and numerical chromosome aberrations is referred to as genomic instability.

A direct consequence of genome instability is intra-tumour heterogeneity \cite{Burrell:2013aa}.
A sample of a biopsy from a genomically unstable cancer will contain tumour cells from lineages that have diverged in the structure and content of their chromosomes.
Significant changes do not usually out-pace cell division.
Thus collections of cells, referred to as clones, will be genomically similar.
Bulk sequencing of such a sample mixes the signals from tumour clones and contaminating normal cells.
An important problem in cancer genomics is the unmixing of these signals, and reconstruction of the structure and content of the genomes of each clone.
The key difficulty of the problem is that mixing dilutes the signal of the changes of interest, often to a level approaching that of the noise in the data.

Existing methods focus on accurate modeling of the number of the copies of each reference genome segment in the sequenced tumour clone or clones.
A simple model of genome structure predominates for most tools: segments in the model are adjacent only if they are also adjacent in the reference genome.
An additional copy of a segment is implicitly modeled as a copied and truncated chromosome, when in reality it may be a tandem duplication resident on the original chromosome.
Theta and Theta2 \cite{Oesper:2012vn,Oesper:2014fj} infer the copy number of tumour clone genomes and mixing proportions of tumour clones and contaminating normal cells.
Both tools assume a-priori knowledge of large segments of the genome with identical clone specific copy number, and model adjacent segments independently.
Titan \cite{Ha:2014fr} uses an HMM to model spacial correlation between segments adjacent in the reference genome, however the state space of the HMM is restricted to allow only one aberrant genotype per segment.
A similar method, CloneHD \cite{Fischer:2014zl} uses an factorial-HMM with a more comprehensive state space.

Simplified models of connectivity are reasonable for genomic profiles using array based technologies.
With whole genome sequencing, tumour specific adjacencies, or breakpoints, are readily available and can be predicted with reasonable accuracy using a variety of tools.
Breakpoints provide the potential for a more comprehensive model of genome structure that includes long range connectivity between genomic segments.
A important question in computational biology is the extent to which a more comprehensive model of genome structure has the potential to improve copy number inference.
Furthermore, a method that integrates both copy number and breakpoints could provide additional information about each breakpoint: whether the breakpoint is real or a false positive, the prevalence of the breakpoint in the clone mixture, and the number of chromosomes harboring the breakpoint per clone.

Some progress has been made on more comprehensive modeling of genome structure in tumour clones.
\cite{Mahmoody:2012ve} proposes an algorithm to infer missing adjacencies in a mixture of rearranged tumour genomes, however they do not model copy number.
\cite{Zerbino:2013pd} proposes a framework for sampling from the rearrangement history of tumour genomes.
\cite{Oesper:2012vn} proposes PREGO, a method for inferring the copy number of segments and breakpoints using a genome graph based approach, though they do not model normal contamination or tumour heterogeneity, limiting applicability of their method to real tumours.

Additional progress in the ability to infer divergent genome structure remains relevant to cancer research.
Subclonal copy changes (changes in a subset of clones) are difficult to assess with single cell methods, whereas coincident subclonal breakpoints could be more easily assessed for their presence in individual cells.
Given multiple samples, an ability to identify subclonal events will enable more accurate tracking of complicated patterns of metastasis, such as sample heterogeneity produced by multiple-metastases to the same anatomic site \cite{Cooper:2015aa}.
Furthermore, subclonal changes represent contemporary events, whereas ubiquitous changes represent historical events.
The ability to discern historical from contemporary breakpoints may provide insight into the repair mechanisms that have been active historically versus those that were active at the time the tumour was sampled.

We propose a method for joint inference of genomic content and structure given tumour sequencing data and a set or predicted breakpoints.
Our method is built upon two important assumptions.
First we assume that intelligent aggregation can be used to gain additional statistical strength, and increased power to detect changes in a minor tumour clone.
Thus, we choose a larger segmentation than competing methods.
To avoid the additional noise associated with a true copy number change occuring in the middle of a segment, we augment our segmentation with breakends of predicted breakpoints, with the intention of capturing the majority of true copy number changes across the genome.
Furthermore, we use counts of reads for alleles of each haplotype block, rather than for alleles of each heterozygous SNP, increasing statistical strength for inference of allele specific copy number changes.
Second, we assume that a more comprehensive model of genome structure that includes the long range connectivity implied by rearrangement breakpoints will improve inference accuracy.
We model the likelihood of observing the sequencing data given a mixture of genome graphs, where each genome graph represents the content and structure of the genomes of tumour clones.
Our method can be considered a natural extension of the factorial-HMM used by cloneHD, which can be thought of as a mixture of HMMs for modeling the clone mixture.

We show using simulated genomes that our method out-performs Titan, Theta2 and CloneHD for inference of clone specific copy number.
We also compare our method against two breakpoint naive methods for inferring copy number of segments and breakpoints: a factorial HMM, and a model assuming indepenence between segments.
For the breakpoint naive methods, we post-hoc assign breakpoint copy number based on segment copy number.
We show that integration of breakpoints improves inference of segment copy number given the same data, and that post-hoc assignment of breakpoint copy number is less accurate than using an integrated model.

\section{Problem Definition}

Assume as given whole genome sequence data from tumour and matched normal samples.
Suppose we have predicted a set of tumour specific rearrangement breakpoints with reasonable accuracy.
We aim to predict the following:
\begin{enumerate}
	\item mixture proportions of tumour clones and normal cells
	\item clone and allele specific copy number of genomic segments
	\item clone specific copy number of rearrangement breakpoints
\end{enumerate}

The strongest signal of copy number change is from segment specific differences in counts of reads aligned concordantly to the reference genome.
We thus model the likelihood of concordant read counts for large segments.
A likelihood model on its own will over-fit to the data, thus we also impose a structure on the clone specific copy number by modeling each clone as a genome graph in a mixture of genome graphs.
A more detailed description of the problem is provided below.

\subsection{Mixtures of Genome Graphs}

Assume as input a set of alignments of uniquely mapped concordant reads, and a set of putative breakpoints predicted from discordant reads.
A segmentation of the genome is not assumed.
We segment the genome using a relatively large segment length (3MB for this study), and augment the set of segment boundaries with breakends of rearrangement breakpoints.
More formally, let $\mathcal{V}$ represent the space of all reference nucleotides, with each nucleotide represented as a chromosome position pair.
Impose an arbitrary order on the chromosomes thus making elements of $\mathcal{V}$ sortable.
Each concordant read alignment $r=(x,y); x, y \in \mathcal{V}$ is a pair of positions representing the start and end of the alignment of the read in the genome.
Assume a set of breakpoints $B$ have been identified from analysis of the discordant reads.
Each breakpoint $b \in B$; $b = (v_j, v_k)$; $v_j, v_k \in W$ is a pair of segment extremities putatively adjacent in a tumour chromosome, but not adjacent in the reference genome.

Let $S$ be a set of $N$ segments defined on a set of $2N$ ordered segment extremities $W \subset \mathcal{V}$.
Segment extremities include boundaries of regular segments and breakends ($\cup_{(v_j, v_k) \in B} \{v_j, v_k\}$).
Segment $s_n \in S$ is defined by a pair of extremities $s_n=(v_{2n-1}, v_{2n})$ such that $v_{2n-1} \leq v_{2n}$.
The length of segment $n$ is thus $l_n = v_{2n} - v_{2n-1} + 1$.
Let $A$ be the set of reference adjacencies, such that $a=(v_{2n}, v_{2n+1}) \in A$ if and only if $v_{2n}$ and $v_{2n+1}$ are on the same chromosomes.
We assume the segmentation covers the entire genome, thus $\forall a \in A; a=(v_{2n}, v_{2n+1}): v_{2n}+1 = v_{2n+1}$.

In construction of the genome graph representation of rearranged tumour genomes we will need an artificial construct to represent tumour specific chromosome ends which we call telomeres \footnote{acknowledging the biological inaccuracy of applying the term telomere}.
Let $s_{N+1} = (v_{2N+1}, v_{2N+2})$ be a dummy telomere segment between two dummy extremities $v_{2N+1}$ and $v_{2N+2}$.
Define a telomere as an adjacency between a real segment extremity and either $v_{2N+1}$ or $v_{2N+2}$.
Define the space of all telomeres $T$ as the edges of a complete bipartite graph between vertex set $W$ and vertex set $\{v_{2N+1}, v_{2N+2}\}$.

Let vertex set $V=W \cup \{ v_{2N+1}, v_{2N+2} \}$ be the set of segment extremities plus 2 additional dummy telomere vertices.
In defining the \emph{genome graph}, we will deliberately use sets of position pairs interchangeably with their representative edges for $A$, $B$, $S$ and $T$.
Define the \emph{genome graph} $H=\left(V,E=(S,Q)\right); Q=A \cup B \cup T$ as a bi-edge-colored graph on vertex set $V$, where \emph{segment edges} in $S$ are given a color distinct from \emph{bond edges} in $Q$ (represented as thick dashed and thin solid respectively in Figure~\ref{fig:genomegraph}).
Bond edges have 3 classes, \emph{reference} (edge set $A$), \emph{breakpoint} (edge set $B$) and \emph{telomere} (edge set $T$).

\begin{figure}[h]
	\centering
	\begin{centering}
		\includegraphics{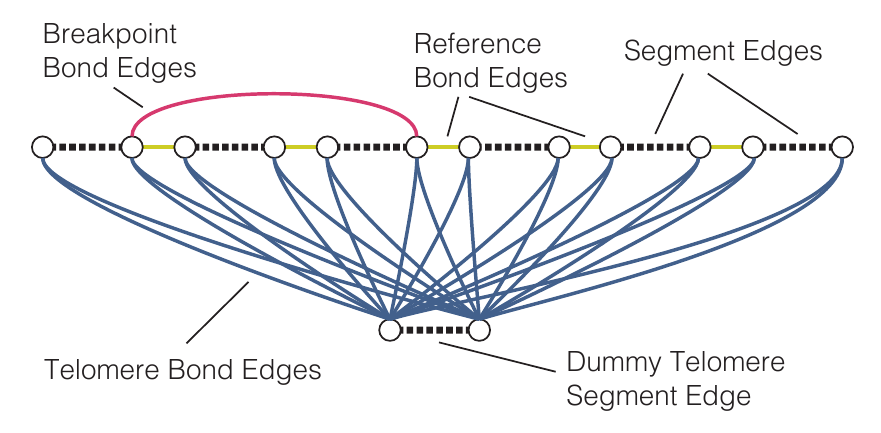}
	\end{centering}
	\caption[Genome Graph]{
		A genome graph on 6 regular segments.
		Segments edges (black, thick dashed) connect vertices representing segment extremities.
		Reference bond edges (yellow, thin solid) connect segments to recapitulate the reference chromosomes.
		Breakpoint bond edges (red, thin solid) represent putative connections between segment ends as identified through analysis of discordant sequencing reads.
		A dummy segment edge (bottom) is used to connect the end points of linear chromosomes via telomere bond edges (blue, thin solid).
		Telomere bond edges form a complete bipartite graph on the pair of vertices incident to the dummy segment edge, and all other vertices.
	}
	\label{fig:genomegraph}
\end{figure}

Define a \emph{linear chromosome} as a sequence of (possibly repeated) oriented segments, and a \emph{circular chromosome} as a cycle of (possibly repeated) oriented segments.
Thus a linear chromosome is an alternating walk in $H$ starting and ending with a segment edge, and a circular chromosome is an alternating tour in $H$.
For algorithmic convenience, we will model all chromosomes as alternating tours.
An alternating walk that represents a linear chromosome can be transformed into a tour by connecting each end of the chromosome to opposite ends of the dummy telomere segment edge using telomere bond edges.
A telomere bond edge is considered \emph{observed} if it is incident to a vertex representing the end of a reference chromosome, with remaining telomere bond edges considered \emph{unobserved}.
All reference and breakpoint bond edges are considered \emph{observed}.

A genome can be represented exactly as a collection of alternating tours in $H$, though such a representation would be unidentifiable based only on adjacency information provided by whole genome sequencing.
Instead, a collapsed representation is more convenient.
Define a \emph{genome instance} $g: E \to \mathbb{N}$ as an assignment of counts to edges in $H$.
A genome instance is \emph{valid} if and only if there exists a collection $\mathcal{T}$ of alternating tours in $H$ for which each edge $e$ appears $g(e)$ times in $\mathcal{T}$.
Alternatively, let $S(v)$ and $Q(v)$ be segment and bond edges incident with vertex $v$ respectively.
A genome instance is valid only if copy number balance condition holds (\cite{Oesper:2012vn}, Equation~\ref{eqn:cnbalance}).
\begin{equation}
\displaystyle
\forall v \in V: \sum_{e \in S(v)} g(e) = \sum_{e \in Q(v)} g(e)
\label{eqn:cnbalance}
\end{equation}

Call a genome instance that obeys the copy number balance condition as \emph{balanced}.
A genome instance $g$ may be balanced but not valid if for some edge $e$, $g(e) < 0$.
Balanced genome instances will be important as modifications of valid genome instances.

Define a \emph{genome collection} $\mathcal{G}$ as a collection of $M$ genome instances.
Let $R_m$ be the number of reads contributed by cell population $m$ in a heterogeneous tumour sample.
We assume each of the $R_m$ reads are sampled uniformly from the length $L_m$ of the genome of cell population $m$.
Thus, each cell population contributes a specific number of reads per nucleotide, $h_m = \frac{R_m}{L_m}$, to the sequencing experiment.
We call $h_m$ the \emph{haploid read depth} since it is the read depth contributed by a single copy of a segment.
Define a \emph{genome mixture} as a pair $(\mathcal{G},h)$, where $\mathcal{G}$ is a collection of genomes and $h$ is a vector of haploid read depths, one per genome.

The expected read count of segment $n$ will thus depend on: 
a) the length of the segment
b) the copy number of the segment in each population
c) the haploid read depth of each population.
Observed read counts will be subject to sampling error, thus we can form a likelihood of expected read counts given observed read counts.

Our goal is to identify a genome mixture that:
\begin{enumerate}
\item maximizes the likelihood of the observed read counts
\item minimizes the number of unobserved telomere bond edges used by each tumour genome
\end{enumerate}

\subsection{Allele Specific Genome Graph}

We model the copy number of each parental allele for each segment.
Complete loss of a parental allele, termed Loss Of Heterozygousity (LOH), is of particular biological interest as such events frequently occur as part of a `double hit' targeting a tumour suppressor gene.
A method that infers the specific copy number of each allele will enable us to identify such biologically important events.
Furthermore, many previous methods have shown increased performance when modelling allele specific versus total copy number.

We define the \emph{allele specific genome graph} $H'$ to jointly model both genome structure and allele specific copy number.
Construct $H'=(V',E')$, $E'=S' \cup Q'$, from genome graph $H$ by duplicating edges and vertices for alleles 1 and 2.
For each vertex $v_n \in V$ create two vertices $v_{n,1}, v_{n,2}$ in $V'$.
For each segment edge $s_n=(v_{2n-1}, v_{2n}) \in S$, create two segment allele edges $s_{n,1}=(v_{2n-1,1}, v_{2n,1})$ and $s_{n,2}=(v_{2n-1,2}, v_{2n,2})$ in $S'$.
For each bond edge $e=(v_j, v_k) \in Q$, create four bond edges $(v_{j,1}, v_{k,1})$, $(v_{j,1}, v_{k,2})$, $(v_{j,2}, v_{k,1})$, $(v_{j,2}, v_{k,2})$ in $Q'$.
Subsequent sections will refer to the allele specific genome graph and its edges and vertices as $H$, $E=S\cup Q$ and $V$ respectively.

\subsection{Haplotype Allele Specific Read Counts}

Heterozygous germline SNPs can be used to classify reads (or previously array probe intensities) into one of the two parental alleles.
Classification into paternal/maternal is impossible without parental DNA, and somewhat irrelevant in the context of this work.
Thus parental alleles and associated read counts are distinguished from each other as \emph{major} and \emph{minor} according to which allele has more copies, or which class of reads has more counts.

Considerable power to estimate allele specific copy number can be gained by using \emph{haplotype} information.
Let the $x_i \in \{0,1\}$ be a binary indicator representing the two possible alleles of heterozygous SNP $i$.
A \emph{haplotype block} $h=(i,k,y), y \in \{0,1\}^k$ is a sequence of alleles $\left( x_i=y_1 ,.., x_{i+k-1}=y_k \right)$ for $k$ consecutive SNPs starting at $i$, where the sequence of alleles exist consecutively on the same physical chromosome.
The \emph{alternate haplotype block} $\bar h = \left( x_i=\bar y_1 ,.., x_{i+k-1}=\bar y_k \right)$ represents the other of the two parental chromosomes (here $\bar x = 1-x$).

Major and minor read counts can be grouped by haplotype block for increased statistical strength.
Specifically, call a read $r$ as \emph{non-conflicting} with $h=(i,k,y)$ if for all $j \in \{ i ,.., i+k-1 \}$ read $r$ matches allele $x_{i+j-1}=y_j$.
Call a read $r$ as \emph{supporting} of $h$ if it is non-conflicting, and contains at least one SNP $j$ from $j \in \{ i ,.., i+k-1 \}$.
Let $z_h$ and $z_{\bar h}$ be counts of reads that support $h$ and $\bar h$ respectively.
Assuming the haplotypes are correct, $z_h$ and $z_{\bar h}$ will provide more accurate representations of allele specific copy number than counts of reads for individual SNPs.
We infer haplotype information using 1000 genomes data and shapeit2 \cite{Delaneau:2012qd}.

\subsection{Genome Mixture Likelihoods}

As described above, assume $\mathcal{G}$ is a collection of $M$ genomes, indexed by $m$, and genome instance $g_m$ is a mapping $g: E \to \mathbb{N}$ assigning a copy number to each edge in the allele specific genome graph $H$.
For convenience, let matrix $c_n$ denote the clone specific copy number state of segment $n$, such that $c_{n m \ell}$ denotes the copy number of segment $n$, genome $m$, allele $\ell$.
In other words, $c_{n m \ell}=g_m(s_{n \ell})$ for segment allele $s_{n \ell}$.

Observed data $X \in \mathbb{N}^{N \times 3}$ are per-segment total and allele specific read counts.
For segment $n$, $x_{n3}$ is the count of concordantly aligning reads contained within segment $n$, and $x_{n1}$ and $x_{n2}$ are the counts of the subsets of those reads classified as from the major or minor allele ($k=1$ and $k=2$ respectively).
As described above, alleles are distinguish as major/minor based on which allele has higher read count, though this does not necessarily relate to which allele has more copies in individual tumour clones.
Modeled alleles $\ell=1$ and $\ell=2$ correspond with measured major ($k=1$) and minor ($k=2$) alleles respectively.

We use haplotype blocks to accurately calculate major and minor read counts.
First, haplotype blocks spanning segment boundaries are split into a separate block per spanned segment.
Second, we calculate supporting read counts $z_{h_{n,i}}$ and $z_{\bar h_{n,i}}$ for each haplotype block $i$ in segment $n$.
Allele specific read counts $x_{n \ell}$ are calculated as given by Equation~\ref{eqn:majorminorcount}.
\begin{eqnarray}
x_{nk} &=&
\begin{cases}
	\displaystyle
	\sum_i \max(z_{h_{n,i}}, z_{\bar h_{n,i}}) &\mbox{if } k = 1 \\
	\displaystyle
	\sum_i \min(z_{h_{n,i}}, z_{\bar h_{n,i}}) &\mbox{if } k = 2
\end{cases} \label{eqn:majorminorcount}
\end{eqnarray}

Let $p_{n \ell k}$ represent, for segment $n$, the proportion of reads from allele $\ell$ that can contribute to measurement $k$, assumed known.
For total read counts ($k=3$), $p_{n k \ell} = 1$ for $\ell \in \{1, 2\}$, since total read count includes all reads from both alleles.
For major and minor read counts, calculate the proportion of reads that can be genotyped by heterozygous SNPs for segment $n$ as given by Equation~\ref{eqn:propallelereads}.
Since reads from allele $\ell=1$ can not contribute to measured read count $x_{n2}$, and visa versa, $p_{n k \ell}=0$ for $k \neq \ell; k \in \{1, 2\}$.
Equation~\ref{eqn:definep} fully specifies $p_{n \ell k}$.
\begin{eqnarray}
\phi_n &=& \frac{x_{n1} + x_{n2}}{x_{n3}} \label{eqn:propallelereads}\\
p_{n k \ell} &=& 
\begin{cases}
	\phi_n &\mbox{if } k = \ell, k \in \{1, 2\} \\
	1 &\mbox{if } k=3 \\
	0 &\mbox{else}
\end{cases} \label{eqn:definep}
\end{eqnarray}

Let $\mu_{nk} \in \mathbb{R}_{> 0}$ model expected major/minor/total read count, calculated as given by Equation~\ref{eqn:expectedreadcount}.
\begin{eqnarray}
\mu_{n k} &=& \sum_m \sum_{\ell} l_n h_m c_{n m \ell} p_{n k \ell} \label{eqn:expectedreadcount}
\end{eqnarray}

We model the likelihood of allele specific or total read counts $x$ given expected read counts $\mu$ using either a Poisson (Equation~\ref{eqn:poissonlikelihood}) or Negative Binomial (Equation~\ref{eqn:nblikelihood}).
The over-dispersion constant for the Negative Binomial is estimated off-line (Section~\ref{sec:overdispersion}).
See Section~\ref{sec:segmentindependence} for a discussion of the independence assumption in comparison to similar models. 
\begin{eqnarray}
p(x | \mu) &=& \frac{ \mu^{x} e^{-\mu} } { x! } \label{eqn:poissonlikelihood} \\
p(x | \mu, r) &=& \frac{\Gamma(r+x)}{\Gamma(x) \Gamma(r)} \left( \frac{r}{r+\mu} \right)^{r} \left( \frac{\mu}{r+\mu} \right)^{x} \label{eqn:nblikelihood}
\end{eqnarray}

\subsection{Prior Probability over Genome Structure}

Positive copy number assigned to a bond edge $e$ by genome instance $g_m$ implies edge $e$ `exists' in the genome mixture.
Observed bond edges are assumed to have higher prior probability of existing.
Thus, we place a prior probability over the number of unobserved edges used by genome instances in a genome mixture.
Let $U \subset Q$ be the set of unobserved bond edges in $H$.
Calculate a prior for the copy number of genome instance $g$ as given by Equation~\ref{eqn:priortelo}.
\begin{eqnarray}
\displaystyle
P(g|\beta) &\propto& \prod_{e \in U} e^{ - \beta g(e) }
\label{eqn:priortelo}
\end{eqnarray}

In log space the above prior amounts to a fixed penalty on each additional copy of an unobserved bond edge.
Such a prior prevents over-fitting of the genome structure, such as a genome instance that assigns positive copy number to many telomere edges in order to be able to fit each segment edge perfectly to the segment likelihood.
Thus the impact of Equation~\ref{eqn:priortelo} is similar to that of a transition matrix in an HMM for modeling spacial correlation, and in fact the genome graph model with no breakpoints has an equivalent representation as an HMM.
Higher values of $\beta$ have the effect of smoothing over false positive deviations in predicted copy number that result from sampling error of identical true copy number.

\subsection{Maximum Posterior Genome Mixtures}

Our overall objective is to identify the genome mixture $(\mathcal{G},h)$ that maximizes the full posterior (Equation~\ref{eqn:fullposterior}) given $\beta$.
\begin{eqnarray}
\displaystyle
p(\mathcal{G}, h | X, \beta) &\propto& p(X | \mathcal{G}, h) \prod_g P(g | \beta)
\label{eqn:fullposterior}
\end{eqnarray}

\section{Method}

\subsection{Method Overview}

We separate the maximum posterior genome mixture problem into two subproblems: 
\begin{enumerate}
	\item learn $h$
	\item infer $\mathcal{G}$ given $h$
\end{enumerate}

For problem 1, we remove breakpoint edges from the genome graph, reducing the graph to a hidden markov model (HMM), making learning of $h$ more tractable.
For problem 2, we use an approximate combinatorial method to infer $\mathcal{G}$ given $h$, initializing at the results of running the viterbi algorithm on the HMM used to learn $h$.

\subsection{Expectation Maximization Method for Learning $h$}

We learn $h$ using the Baum Welch algorithm for learning the parameters of a Hidden Markov Model (see Section~\ref{sec:learnh}).
In brief, we use expectation maximization to find a value a local maxima of the marginal likelihood function (Equation~\ref{eqn:marginallikelihood}).
\begin{eqnarray}
L(X|h) &=& \sum_C p(X,C|h)
\label{eqn:marginallikelihood}
\end{eqnarray}

At iteration $t$, the algorithm calculates the $h^{(t)}$ that maximizes the expected value of the complete data log likelihood $p(X,C|h)$ with respect to the conditional distribution $p(C|X,h^{(t-1)},\cdot)$.
We use multiple restarts with different initial $h^{(0)}$ in an attempt to discover the global maximum.

\subsection{Combinatorial Method for Inferring $\mathcal{G}$}

We propose a greedy algorithm for inferring the full structure of the genome collection $\mathcal{G}$ given $h$ estimated from large segments.
Given a current solution $\mathcal{G}^{(t)}$, our aim is to select from a set of possible modifications that a) are simple to calculate, b) are comprehensive enough to escape local optima, c) produce a valid genome collection $\mathcal{G}^{(t+1)}$ when applied to $\mathcal{G}^{(t)}$.
A valid genome collection is defined as a genome collection for which all genomes $g \in \mathcal{G}$ have positive copy number and obey the copy number balance condition.

Modeling the likelihood of total reads introduces a dependency between the copy number of each allele of the same segment, complicating inference.
For the purposes of inferring $\mathcal{G}$, we model only the likelihood of allele specific read counts, allowing for the independent modeling of the alleles of each segment.
Let $c_e$ be a vector of per clone copy numbers assigned to edge $e$ for each of the $m$ genome instances, such that $c_{em}=g_m(e)$.
Suppose edge $s_{n \ell} \in S$ represents segment $n$, allele $\ell$.
The likelihood of allele specific read counts $x_{nk}$ where $k=\ell$ can be written as $p(x_{nk} | \mu_{n \ell})$ where $\mu_{n \ell}$ is calculated as given by Equation~\ref{eqn:expectedallelereadcount}.
\begin{eqnarray}
\mu_{n \ell} &=& \sum_m l_n h_m c_{em} \phi_n \label{eqn:expectedallelereadcount}
\end{eqnarray}

We define the objective function of our greedy heuristic as given by Equation~\ref{eqn:logpost}.
\begin{eqnarray}
f(\mathcal{G}) &=& - \log p(\mathcal{G}|X,h,\cdot) \nonumber \\
&=& \sum_{e \in E} f_e(c_e) \\
f_e(c_e) &=&
\begin{cases}
- \log p(x_{nk} | \mu_{n \ell}) &\mbox{if } e = s_{n \ell} \in S \\
\sum_m \beta \cdot g_m(e) &\mbox{if } e \in U
\end{cases} \label{eqn:logpost}
\end{eqnarray}

\subsubsection{Representing modifications as edge disjoint sets of alternating cycles}

We propose an algorithm for identifying a modification of $\mathcal{G}^{(t)}$, optimal with respect to $f$, from a restricted set of modifications that increase or decrease the copy number of any edge by at most 1.
Consider valid genome instance $g_1$ and balanced genome instance $g_{\Delta}$.
The set of balanced genome instances is closed under addition and subtraction.
Thus $g_2 = g_1 + g_{\Delta}$ will be a balanced genome instance, and if $g_2(e) = g_1(e) + g_{\Delta}(e) \geq 0 \ \forall e$, then $g_2$ will also be valid.

An obvious set of candidate $g_{\Delta}$ include all $g_{\Delta}$ that set $g_{\Delta}(e) = 1$ (or $g_{\Delta}(e) = -1$) for all edges $e$ in some alternating cycle of $H$.
Such $g_{\Delta}$ would be analogous to adding (or subtracting) a circular chromosome to the genome instance.
However, it is trivial to show for such a restricted set of modifications, a greedy algorithm would easily get stuck in local optima.
In Figure~\ref{fig:modification}A, for instance, transformation of the genome instance on the left to that on the right would not be possible without an intermediate state that assigns 2 copies to each segment edge.
Such an intermediate state may be prohibitively unlikely if there is considerable evidence for the segment to be copy number 1.

\begin{figure}[h]
\centering
\includegraphics{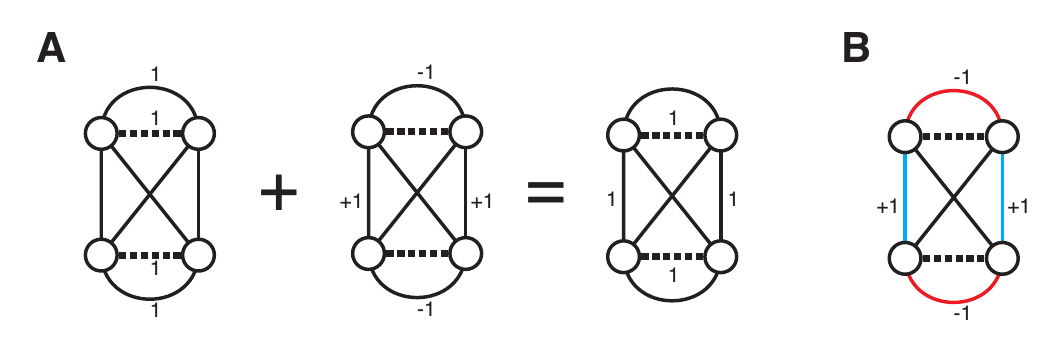}
\caption[Genome Graph Modification]{
Genome graph modifications.  Thick dashed lines show segments edges, thin solid lines show bond edges.
A) Transformation of one genome instance to another via addition of a modification.  Edges are annotated with copy number for edges with non-zero copies.  Note that a transformation performed in two steps, one addition and one subtraction, would likely result in a highly improbable intermediate state with segment edge copies set to 2.
B) Representation of the modification as an alternating cycle in the genome graph.
}
\label{fig:modification}
\end{figure}

Instead, we require candidate $g_{\Delta}$ that both add an subtract edges, as swapping one set of edges for another set will allow easier movement through the space of possible genome collections.
Let $g_+$ be a balanced genome instance constructed such that $g_+(e) = 1$ for all edges in a set of edge disjoint alternating cycles in $H$.
Note that since each vertex is incident to a single segment edge, edge disjoint alternating cycles are also vertex disjoint.
Let $g_-$ be defined similarly to $g_+$, with $g_-(e) = -1$ for a distinct set of edge disjoint alternating cycles in $H$.
Let $g_{\Delta} = g_+ + g_-$.
Call $g_{\Delta}$ thus constructed as a \emph{simple genome modification}.

\begin{proposition}
Any simple genome modification has a representation as an edge disjoint set of alternating cycles in a specific bi-edge-colored graph derived from $H$.
\label{prop:simplemodrep}
\end{proposition}
\begin{proof}
Edges in $g_{\Delta}$ will have copy number in $\{-1, 0, +1\}$.
Remove all 0 edges from $g_{\Delta}$.
Assign a color to each remaining edge as follows:

\begin{itemize}
\item segment edge $e$ with $g_{\Delta}(e) = +1$ $\to$ red
\item segment edge $e$ with $g_{\Delta}(e) = -1$ $\to$ blue
\item bond edge $e$ with $g_{\Delta}(e) = +1$ $\to$ blue
\item bond edge $e$ with $g_{\Delta}(e) = -1$ $\to$ red
\end{itemize}

Now consider any vertex $v \in H$.
Let $s$ be the segment edge incident with $v$, and let $e_i$ be the $i^{\text{th}}$ bond edge incident with $v$.
Suppose $g_{\Delta}(s) = 0$, but there exists $i$ such that $g_{\Delta}(e_i) \neq 0$.
Then there must exist $i_1, i_2$ such that $g_{\Delta}(e_{i_1}) = +1$ and $g_{\Delta}(e_{i_1}) = -1$, otherwise $g_+$ and $g_-$ would not both be vertex disjoint.
Furthermore, $e_{i_1}$ will be colored blue and $e_{i_2}$ will be colored red.
Suppose $g_{\Delta}(s) = +1$.
Then there exists a single bond edge $e_{i}$ incident with $v$ such that $g_{\Delta}(e_{i}) = +1$ by the copy number balanced condition and the requirement that $g_+$ be vertex disjoint.
Edge $s$ will be colored red and $e_{i}$ will be colored blue.
A similar analysis for a segment edge $s$ with $g_{\Delta}(s) = -1$ shows that $s$ will be colored blue and a single incident bond edge will be colored red.
Thus at any vertex, either 0 or 2 incident edges will have non-negative copy number.
Furthermore, for vertices with 2 incident non-negative copy number edges, the color of the edge as defined above will alternate.
Thus $g_{\Delta}$ can be represented as a set of vertex disjoint alternating cycles.
\end{proof}

We now show how to construct a graph $F$ called the \emph{genome modification graph} (Figure~\ref{fig:modgraph}), such that any vertex disjoint alternating cycle in $F$ represents a valid simple genome modification.
Add a single copy of each vertex in $H$ to $F$ and two copies of each edge in $H$ to $F$.
Label one duplicated edge as `+1' and the other duplicated edge as `-1'.
Color each edge as follows:
\begin{itemize}
\item segment edge $e$ labeled $+1$ $\to$ red
\item segment edge $e$ labeled $-1$ $\to$ blue
\item bond edge $e$ labeled $+1$ $\to$ blue
\item bond edge $e$ labeled $-1$ $\to$ red
\end{itemize}

\begin{figure}[h]
\centering
\includegraphics{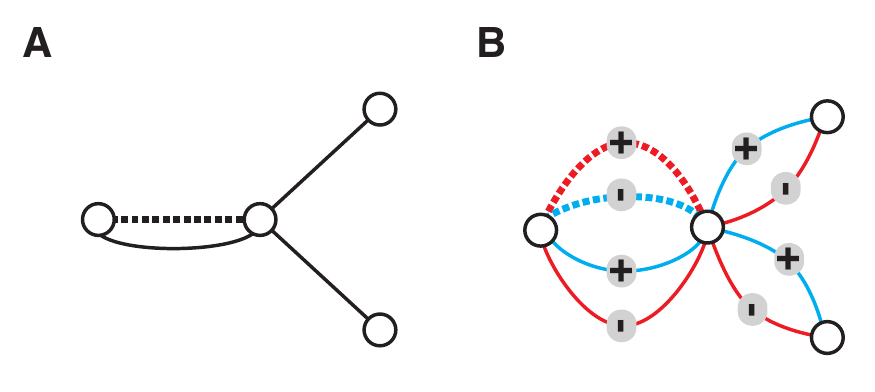}
\caption[Genome Modification Graph]{
A) A small subgraph of a genome graph with a subset of edges is shown.  
B) Transformation of the genome graph to a genome modification graph.  Each edge is duplicated and labelled with `+' or `-'.  Segment edges are colored red for label `+' and blue for label `-'.  Bond edges are given an opposite labelling: blue for label `+' and red for label `-'.  
}
\label{fig:modgraph}
\end{figure}

\begin{proposition}
An edge disjoint set of alternating cycles $\mathcal{C}$ in $F$ represents a valid simple genome modification.
\end{proposition}
\begin{proof}
A transition from $+1$ labeled edges to $-1$ labeled edges (or visa versa) in $\mathcal{C}$ only occurs for two adjacent bond edges.
Let $W$ be the set of vertices with 2 incident bond edges from $\mathcal{C}$.
The size of $W$ will be even since the number of transitions must be even.
Furthermore, it is always possible to construct a set of edge disjoint paths $\mathcal{P}$ connecting pairs of vertices in $W$.
Let $v$ be a vertex in $W$, and $s=(v,u)$ the incident segment edge.
If $u \in W$, then $s$ is a one edge path connecting $v,u \in W$.
If $u \not \in W$, then there must exist another vertex $v' \in W$ with incident segment edge $s'=(v',u')$, and $u' \not \in W$.
Then $(u,u')$ is a bond edge not in $\mathcal{C}$ and $(v,u), (u,u'), (u', v')$ is a three edge alternating path connecting $v,u \in W$.
Now edges in $\mathcal{C}$ labeled $+1$ can be seen as edges unique to a genome instance $g_+$, edges in $\mathcal{C}$ labeled $-1$ can be seen as edges unique to a genome instance $g_-$, and edges in $\mathcal{C}$ are edges common to both $g_+$ and $g_-$.
\end{proof}

\begin{proposition}
Any simple genome modification can be represented as an edge disjoint alternating cycle in $F$.
\end{proposition}
\begin{proof}
Follows directly from Proposition~\ref{prop:simplemodrep}
\end{proof}

Note that simple genome modifications involving only bond edges are equivalent to a balanced rearrangement in a breakpoint graph.

\subsubsection{Selecting an optimal simple genome modification}

For a collection of $M$ genomes, consider the set of edge copy number modifications $\mathbb{T}^M$, where $\mathbb{T}=\{-1, 0, +1\}$.
Let $\Delta \in \mathbb{T}^M, \Delta \neq 0$ be a specific modification affecting at least one of the $M$ genomes.
Let $z \in \mathbb{T}^E$ represent the acceptance of the modification or its inverse for each edge in $E$.
Thus $z_e \Delta$ represents modification by $\Delta$, $-\Delta$ or $0$ for $z_e$ equal to $+1$, $-1$ and 0 respectively.
Write the objective of a locally optimal modification given $\Delta$ as shown in Equation~\ref{eqn:modifyobj}.
\begin{eqnarray}
&&\operatornamewithlimits{argmin}_{z \in \mathbb{T}^E} \sum_{e \in E} f_e (c_e + z_e \Delta)
\label{eqn:modifyobj} \\ 
&\text{s.t.}& \sum_{e \in S(v)} z_e \Delta = \sum_{e \in Q(v)} z_e \Delta\ \ \ \forall v \in V, 
\label{eqn:modifyobjconstraint1} \\
&& 1 \geq \sum_{e \in S(v)} |z_e| \ \ \ \forall v \in V,
\label{eqn:modifyobjconstraint2} \\
&& 1 \geq \sum_{e \in Q(v)} |z_e| \ \ \ \forall v \in V
\label{eqn:modifyobjconstraint3}
\end{eqnarray}

The first constraint (Equation~\ref{eqn:modifyobjconstraint1}) is simply the copy number balance condition.
The second and third constraints (Equations~\ref{eqn:modifyobjconstraint2} and \ref{eqn:modifyobjconstraint3}) ensure that the set of alternating cycles represented by $z$ is vertex disjoint, and is thus a simple genome modification.

To identify the locally optimal modification, first create genome modification graph $F$, then create cost function $\kappa$ assigning cost to each edge as follows:
\begin{itemize}
\item edge $e$ labeled $+1$ $\to$ $f_e(c_e + \Delta) - f_e(c_e)$
\item edge $e$ labeled $-1$ $\to$ $f_e(c_e - \Delta) - f_e(c_e)$.
\end{itemize}

\begin{proposition}
Assume $f_e$ is convex. 
A minimum cost vertex disjoint set of alternating cycles $C$ in $F$ given cost function $\kappa$ is a locally optimal modification minimizing Equation~\ref{eqn:modifyobj}.
\end{proposition}
\begin{proof}
Any $z$ can be transformed into a specific $C$ by selecting the edges in $F$ corresponding to settings of $z$ (label $+1$ if $z_e=+1$, label $-1$ if $z_e=-1$, or no edge if $z_e=0$).
Conversely, suppose $C$ includes both the $+1$ labeled edge and $-1$ labeled edge that correspond to a single edge in $H$.
The convexity of $f_e$ implies that $f_e(c_e + \Delta) - f_e(c_e) + f_e(c_e - \Delta) - f_e(c_e) \geq 0$.
Selection of both edges implies $f_e(c_e + \Delta) - f_e(c_e) + f_e(c_e - \Delta) - f_e(c_e) \leq 0$, otherwise neither would have been selected.
Thus $f_e(c_e + \Delta) - f_e(c_e) + f_e(c_e - \Delta) - f_e(c_e) = 0$, and both edges can be removed from $C$ without consequence to the objective.
Let $C'$ be derived from $C$ by removing pairs of $+1$ labeled edge and $-1$ labeled edges edge corresponding to a single edge in $H$.
There is a bijection between the space of possible $z$ and the space of possible $C'$.
Furthermore, minimizing Equation~\ref{eqn:modifyobj} is equivalent to minimizing $\sum_{e \in E} \left( f_e (c_e + z_e \Delta) - f_e (c_e) \right)$.
The result follows.
\end{proof}

To identify the minimum cost vertex disjoint set of alternating cycles in the bi-edge-colored graph $F=(V, E=(E', E'')))$, we use the following well known transformation.
Create two separate graphs on distinct vertices, $F'=(V', E')$ and $F''=(V'', E'')$, by selecting a subset of edges from $F$.
Merge $F'$ and $F''$ to create $F'''$ by adding additional edges $E'''$ connecting vertices in $F'$ with vertices in $F''$ that originated from the same vertex in $F$.
Let $M$ be a prefect matching in $F'''$.
Such a matching exists since $E'''$ is a perfect matching.
Furthermore, $M$ is a perfect matching in $F'''$ if and only if it is composed of a set of edge disjoint alternating cycles in $F$ and a subset of edges from $E'''$ matching the remaining unmatched vertices in $F'''$.
The set of edge disjoint alternating cycles is the minimum cost such set in $F$.
We use Blossom V \cite{kolmogorov2009blossom} to calculate the minimum cost perfect matching.

\subsubsection{Greedy Algorithm}

We propose the following algorithm for solving the maximum posterior genome mixture problem.
Given are haploid read depths $h$.
Initialize $\mathcal{G}$ to a valid genome collection as follows.
Set the copy number of segment edge $n$ to $c_n$.
Set the copy number of the bond edge parallel to segment edge $n$ also to $c_n$.
The algorithm proceeds as follows.
Given current iteration $t$, identify the minimum cost modification $\Delta$ of $\mathcal{G}^{(t)}$ from the set of all possible modifications $\Delta \in \mathbb{T}^M, \Delta \neq 0$.
If the minimum cost modification has cost less than 0, apply it to $\mathcal{G}^{(t)}$ to create $\mathcal{G}^{(t+1)}$ and continue.
If all modifications have cost 0, stop iteration.

\section{Results}

\subsection{Simulating rearranged genomes}

We developed a principled method of simulating rearranged genomes that fulfilled two important criteria.
First, the simulated tumour genomes are required to have been produced by a known evolutionary history composed of duplication, deletion, and balanced rearrangement events applied successively to a initially non-rearranged normal genome.
Second, the copy number profile of the simulated tumour genome should be reasonably similar to that of previously observed tumours.

Naively applying random rearrangements to a normal genome would result in a significant number of regions that are homozygously deleted.
Such a scenario is unrealistic given that large scale homozygous deletion would remove housekeeping genes necessary for the survival of the cell.
Alternatively, a naively simulated genome could become exceptionally large, an outcome that is rarely observed given the presumed burden of replicating such a genome.
Thus, we developed a re-sampling method for producing realistic rearranged genomes.
At each step in the simulation of an evolutionary history, we re-sample a swarm of genomes according to a \emph{fitness} function.
Fitness is calculated as the multinomial likelihood of the simulated copy numbers given average copy number proportions from a set of real tumours.
For this study we used copy number proportions measured from 7 high grade serous tumours (data not shown).

To simulate a mixture of related genomes, we first simulate the rearrangement history of the common ancestor.
We then modify the fitness function to include term that controls the deviation between ancestor and descendant.
A target deviation is specified as the proportion of the genome with divergent copy number state between ancestor and descendant.
The deviation term is the squared error between simulated deviation and target deviation, with some user specified scaling factor, or variance. 

We simulated 20 mixtures of rearranged tumour genomes.
Genomes in each mixture harbored 50 ancestral and 40 clone specific rearrangements, with an additional 50 false rearrangements.
Each genome consisted of 1000 segments with randomly sampled lengths totaling $3\times10^9$ nt.
Segment lengths as a proportion of genome length sampled from a Dirichlet distribution with concentration parameter 1.
Proportion genotypable reads uniformly from between 0.05 and 0.2.
We assumed the samples were composed of 40\% normal cells and 2 tumour clones.
Target deviation was set at 30\%.
Minor clone proportions were set to 5, 10, 20, and 30\% of cells.
Read counts were simulated using a negative binomial likelihood given segment copy numbers and assuming 40X sequencing \footnote{total haploid coverage of 0.2 reads (paired) per nucleotide corresponds to a 40X sequence coverage genome, 100X100bp reads from $\sim$400bp fragments}.

\subsection{Benchmarking learning haploid depth using simulated data}

We used EM to infer the haploid read depths ($h$ parameter) within the context of the HMM version of our model for simulated datasets.
For 50\% of the simulated datasets, minor clone proportion predictions are within 5\% of the simulated value, and normal proportion predictions are within 2\% of the simulated value (Figure~\ref{fig:learnresults}).

\begin{figure}[h]
\caption[Simulations]{
	Box plots of inferred clonal fraction for the minor tumour clone (left) and normal clone (right).  Datasets are grouped by simulated minor clone proportion (0.05, 0.1, 0.2, 0.4).
	Normal proportion is 0.4 for all simulations.
}
\label{fig:learnresults}
\begin{subfigure}[b]{\textwidth}
	\centering
	\caption{Proportion minor clone in the mixture with normal fixed at 0.4}
	\includegraphics[width=0.5\textwidth]{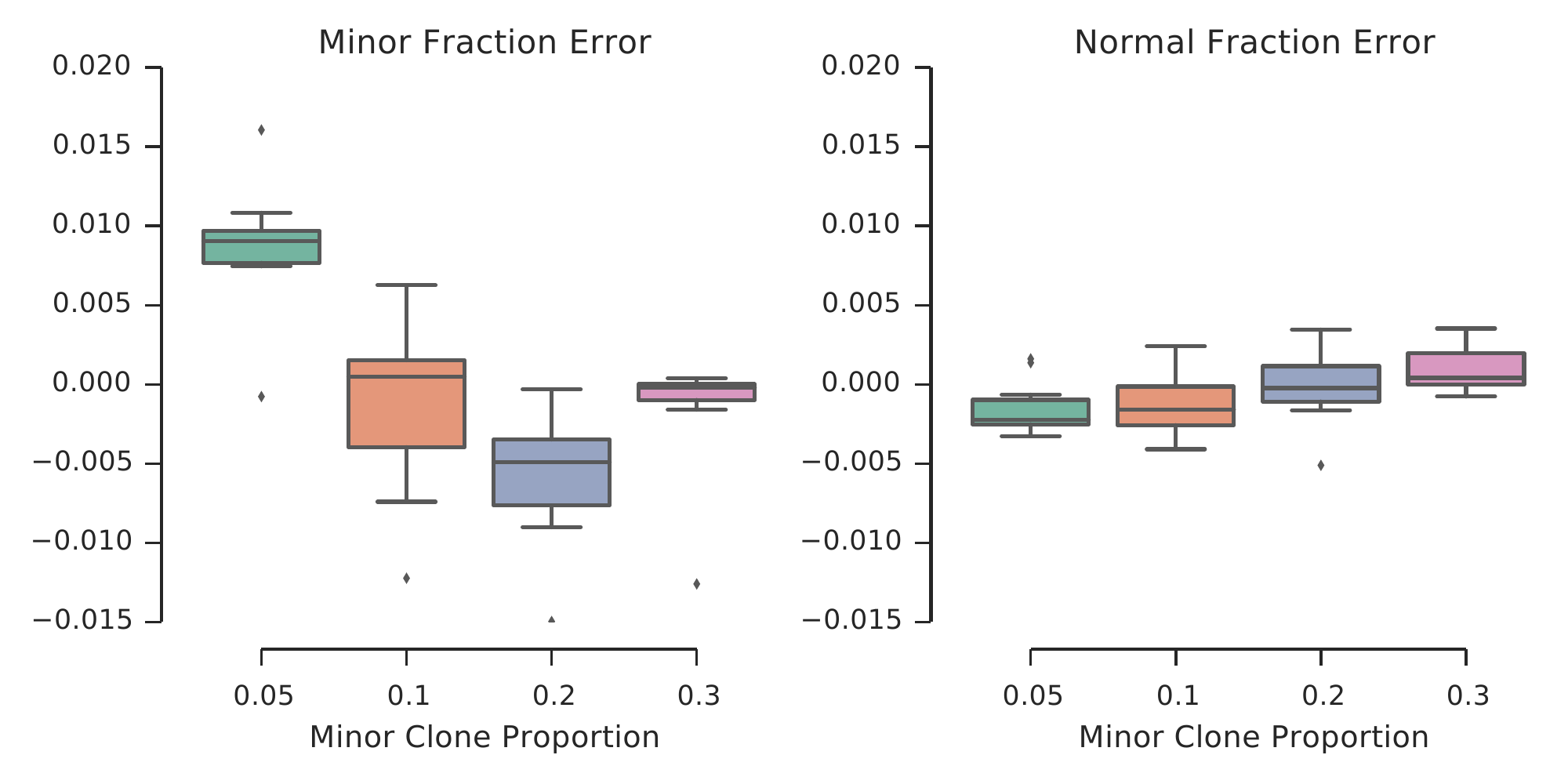}
	\label{fig:learnminorclone}
\end{subfigure}
\begin{subfigure}[b]{\textwidth}
	\centering
	\caption{Proportion of normal in the mixture}
	\includegraphics[width=0.5\textwidth]{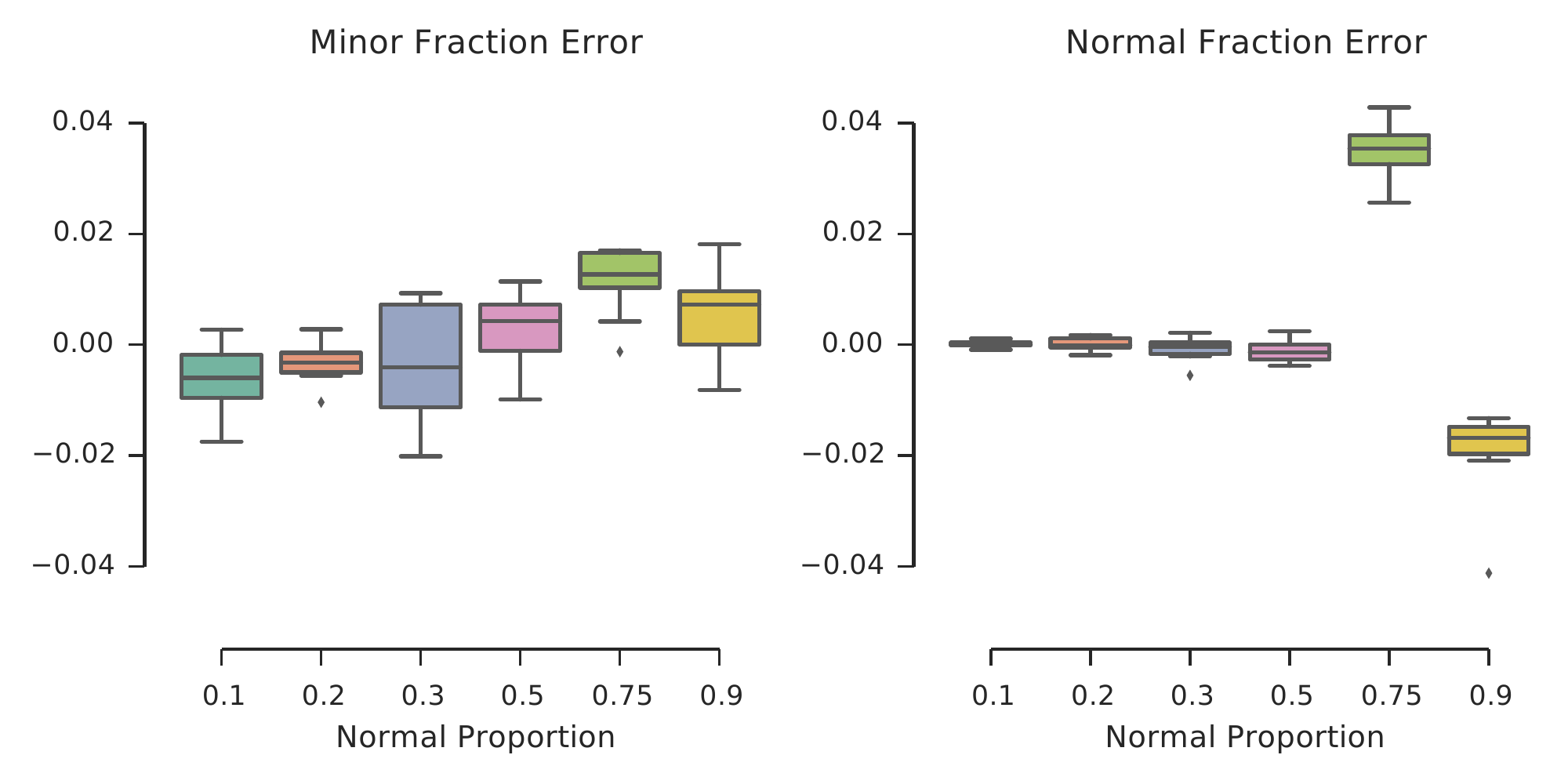}
	\label{fig:learnnormal}
\end{subfigure}
\begin{subfigure}[b]{\textwidth}
	\centering
	\caption{Proportion of genomic segments with divergent (subclonal) copy number}
	\includegraphics[width=0.5\textwidth]{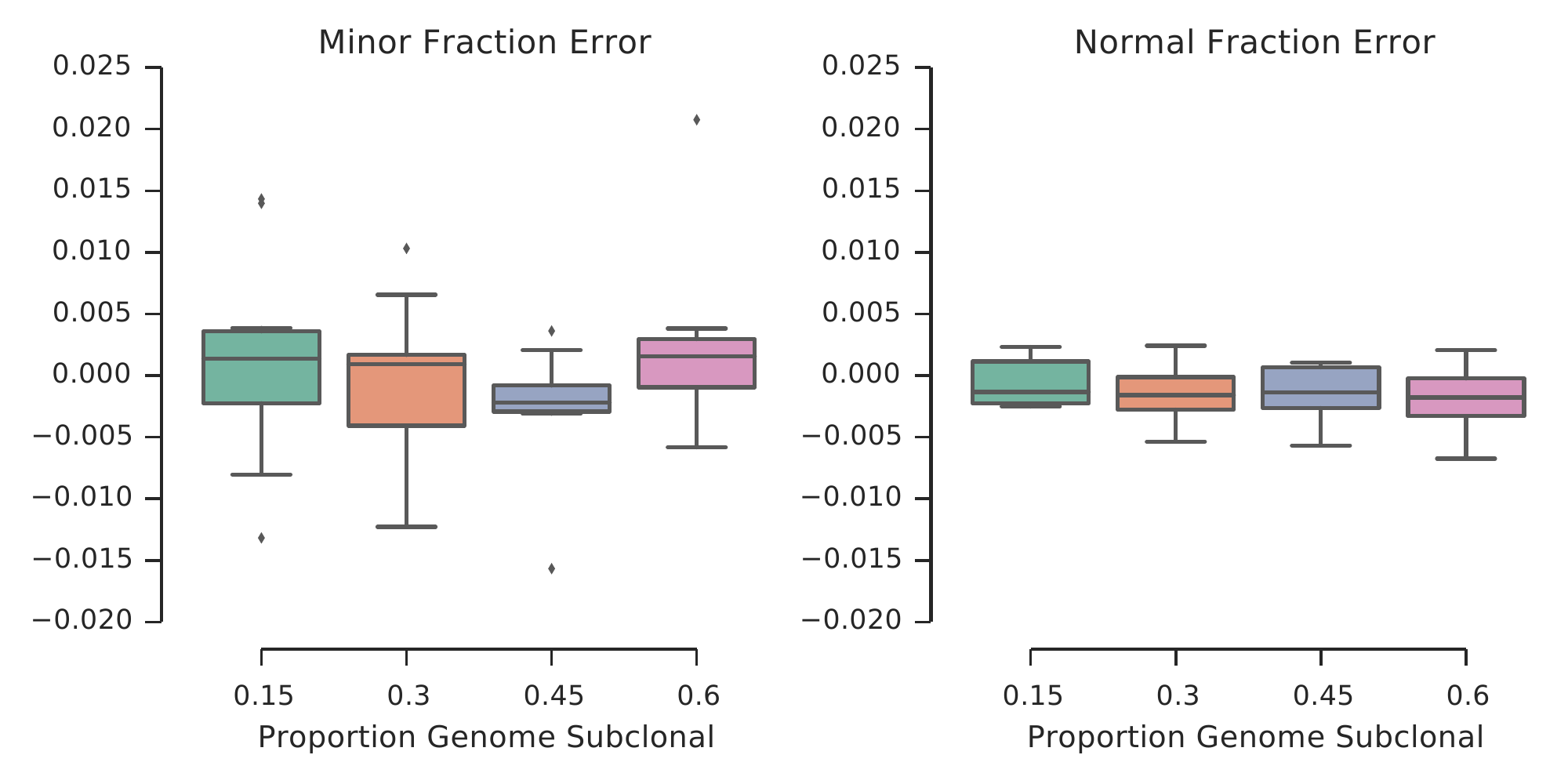}
	\label{fig:learnpropsubclonal}
\end{subfigure}
\end{figure}

\subsection{Benchmarking structure and content prediction using simulated data}

We applied 3 decoding algorithms to the read count data assuming the clone proportions and sequencing depth were known.
The \emph{independent} algorithm calculates the maximum likelihood copy number state of each segment and then post-hoc assigns copy number to breakpoints.
The \emph{Viterbi} algorithm calculates the maximum posterior path through the HMM representation of each chromosome, also assigning breakpoint copy number post-hoc.
The \emph{genomegraph} algorithm uses the proposed algorithm to simultaneously infer segment and breakpoint copy number.
For optimal performance, the genomegraph algorithm is initialized with the results of the Viterbi.

\begin{figure}[h]
\caption{Performance of the genomegraph algorithm compared to two breakpoint naive approaches, the independent model and and HMM using the viterbi algorithm.  For each set of plots, one parameter of the simulation was varied.}
\label{fig:inferenceresults}
\begin{subfigure}[b]{\textwidth}
	\centering
	\caption{Proportion minor clone in the mixture with normal fixed at 0.4}
	\includegraphics[width=0.75\textwidth]{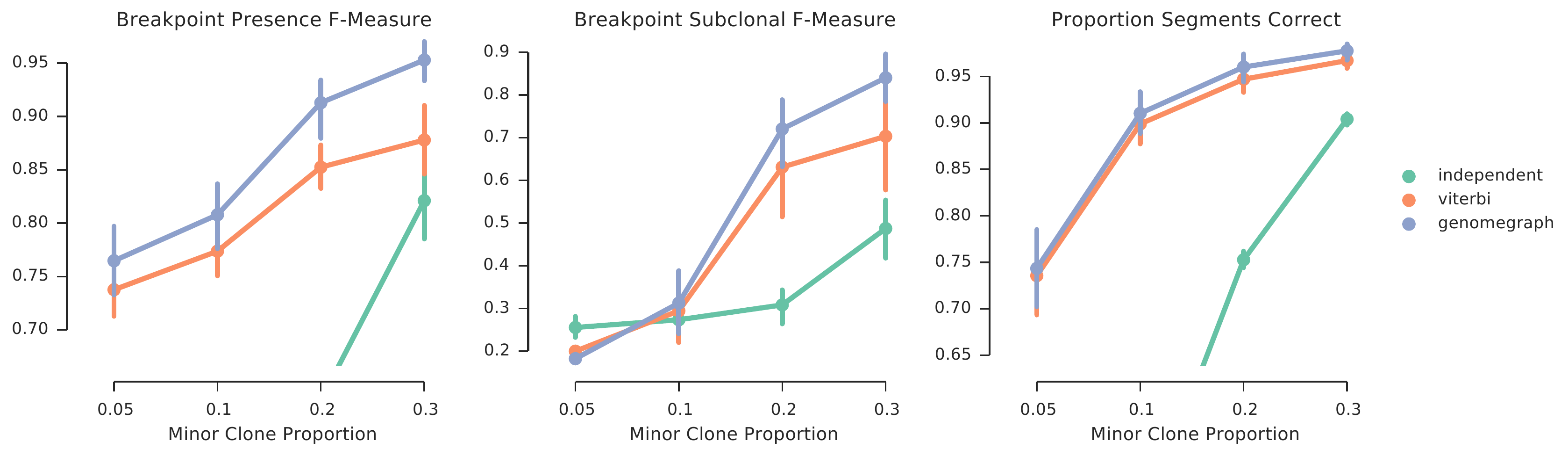}
	\label{fig:inferminorclone}
\end{subfigure}
\begin{subfigure}[b]{\textwidth}
	\centering
	\caption{Proportion of genomic segments with divergent (subclonal) copy number}
	\includegraphics[width=0.75\textwidth]{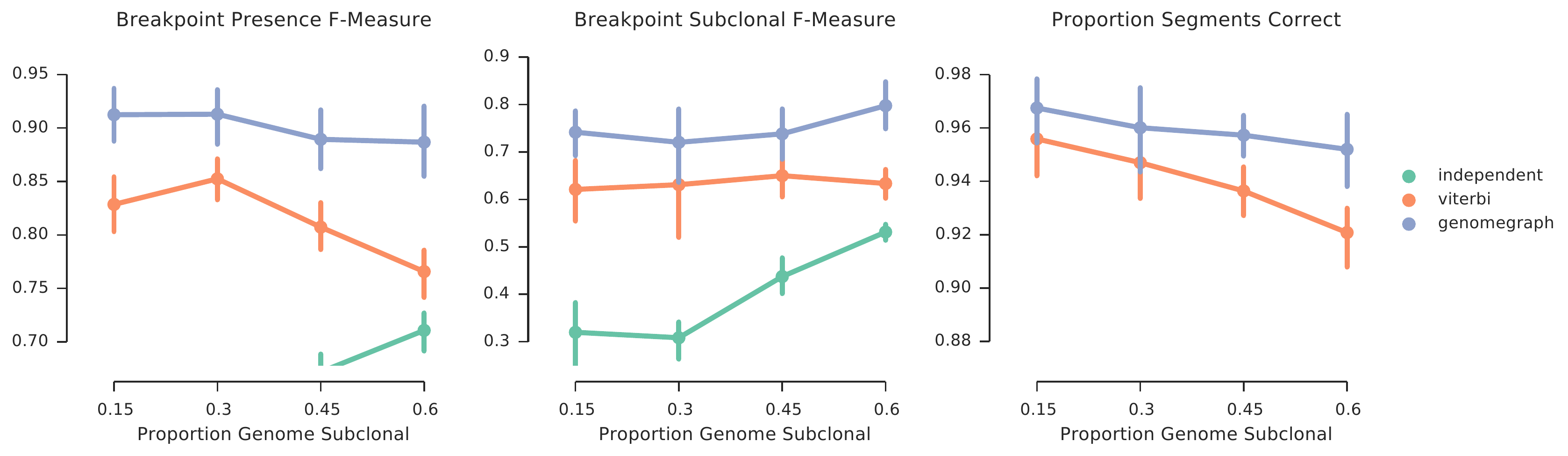}
	\label{fig:inferpropsubclonal}
\end{subfigure}
\begin{subfigure}[b]{\textwidth}
	\centering
	\caption{Proportion of normal in the mixture}
	\includegraphics[width=0.75\textwidth]{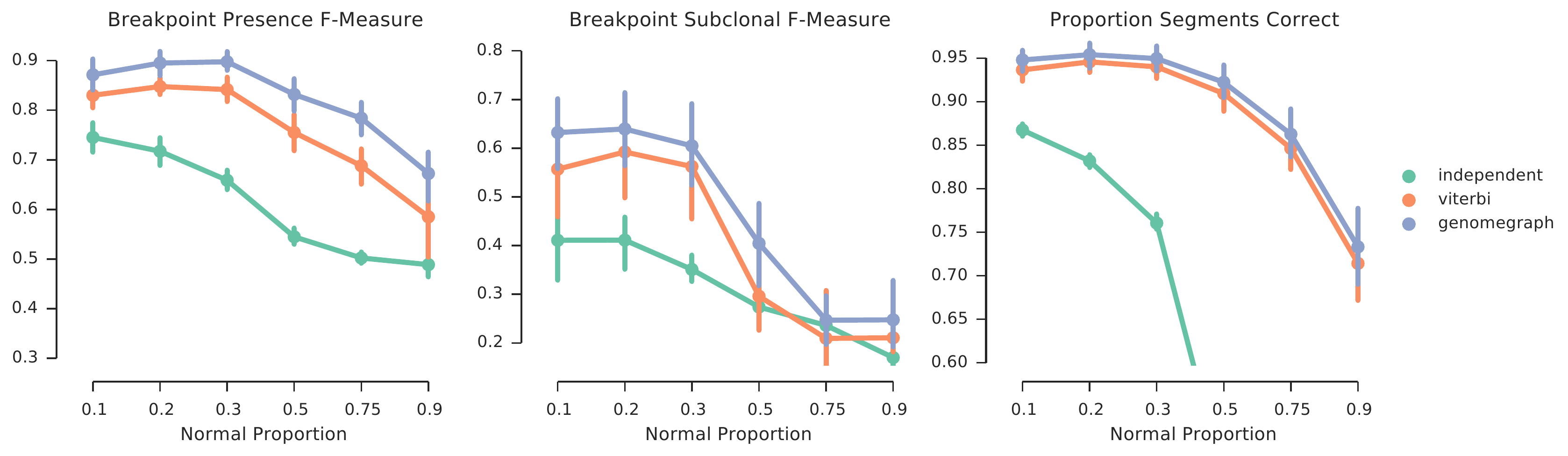}
	\label{fig:infernormal}
\end{subfigure}
\begin{subfigure}[b]{\textwidth}
	\centering
	\caption{Learning error as standard deviation from true fraction}
	\includegraphics[width=0.75\textwidth]{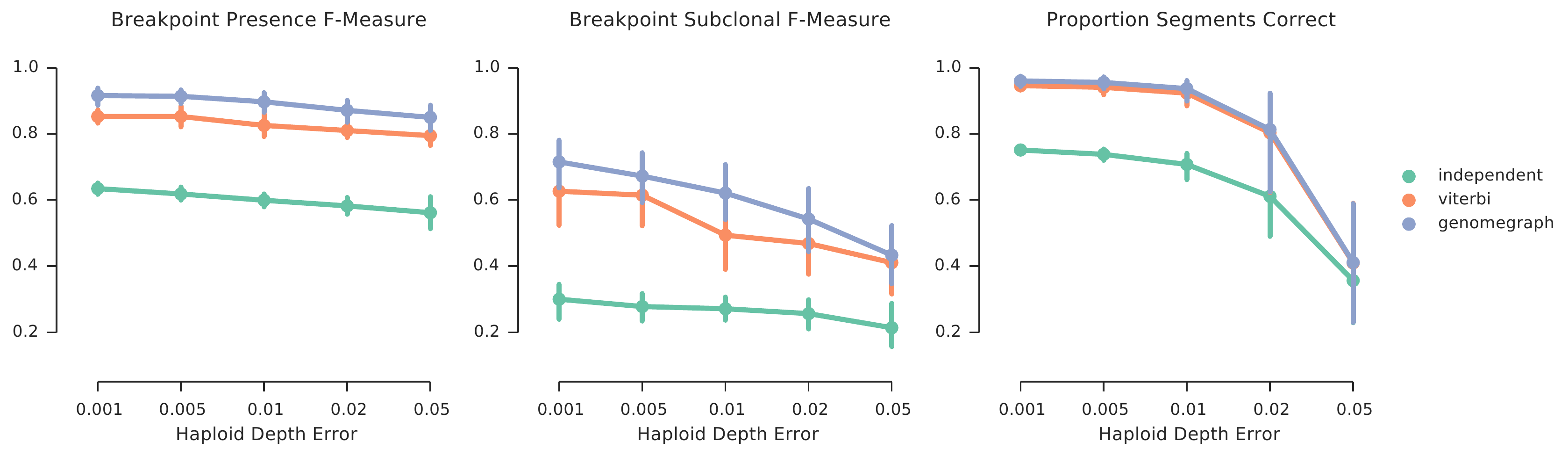}
	\label{fig:inferfracerror}
\end{subfigure}
\end{figure}

We calculated 3 measures of performance, f-measure of ability to predict breakpoints as present versus false, f-measure of ability to predict breakpoints as subclonal, and proportion of segments for which the correct copy number is identified.
The genomegraph algorithm outperforms the independent and Viterbi algorithms by all measures for all but the 5\% minor clone mixtures (Figure~\ref{fig:inferenceresults}).
The inability of the independent algorithm to model spatial correlation between adjacent segments results in a higher number of spurious copy number transitions and low precision with respect to estimation of breakpoint presence and clonality.
Precision is higher for the Viterbi due to the smoothing properties of the algorithm.
However, recall is lower than the genomegraph method, since copy number changepoints either do not precisely coincide with respective breakpoints, or are smoothed over entirely for copy number changes in low proportion clones.
Finally, joint inference noticeably improves the accuracy of segment copy number prediction over the current state of the art, viterbi inference in an HMM.

\subsection{Comparison with Existing Copy Number Inference Methods}

We compared our genome graph approach to four existing methods for subclonal copy number inference in heterogeneous tumour samples: TITAN \cite{Ha:2014fr}, CloneHD \cite{Fischer:2014zl}, and Theta2.0 \cite{Oesper:2012vn,Oesper:2014fj} (see supplementary section~\ref{sec:runtools}).
We first simulated 10 normal genomes, including SNP genotypes.
Recombination sites were selected at a rate of 20 per 100Mb.
For each region between adjacent recombination sites we selected a random individual from a phased 1000 genomes reference panel \cite{Delaneau:2012qd}, and used that individual's SNP genotypes for that region.
For each of the 10 normal genomes, we then simulated 4 pairs of tumour clones using the re-sampling method described above, varying the proportion of the genome that is divergent.
Finally we simulated 4 genome mixtures for each of the 40 pairs of tumour clones.
For each simulated mixture, we calculated the number of reads for each clone in an approximately 30X sequencing dataset, and simulated fragments with mean 300 and standard deviation 30.
We then added added the germline SNPs to 100bp reads from these fragments, flipping the genotype of SNPs based on a simulated sequencing error rate of 0.005.

We compared each tool's output CNA predictions and mixture prediction to the true simulated values using 3 measures
\begin{enumerate}
	\item Proportion of segments with correct major/minor copy number for each clone, or total copy number where allele specific copy number was unavailable
	\item Relative error in estimation of normal proportion.
	\item Relative error in estimation of minor clone proportion.
\end{enumerate}
The results are shown in Figure~\ref{fig:comparisonresults}.
For our simulated data, our method outperforms the competing methods on all measures.
In general, failure of the competing methods to correctly predict copy number is primarily due to the inability to identify the true mixture.
Given an incorrect estimate of the normal contamination and minor clone proportion, copy number estimates are unlikely to be correct.

\begin{figure}[h]
	\caption{Performance of the genomegraph algorithm compared to three existing methods, TITAN, CloneHD, and Theta2.0.  Each plot shows performance with one parameter of the simulation varied.}
	\label{fig:comparisonresults}
	\begin{subfigure}[b]{\textwidth}
		\centering
		\caption{Proportion minor clone in the mixture with normal fixed at 0.4}
		\includegraphics[width=0.75\textwidth]{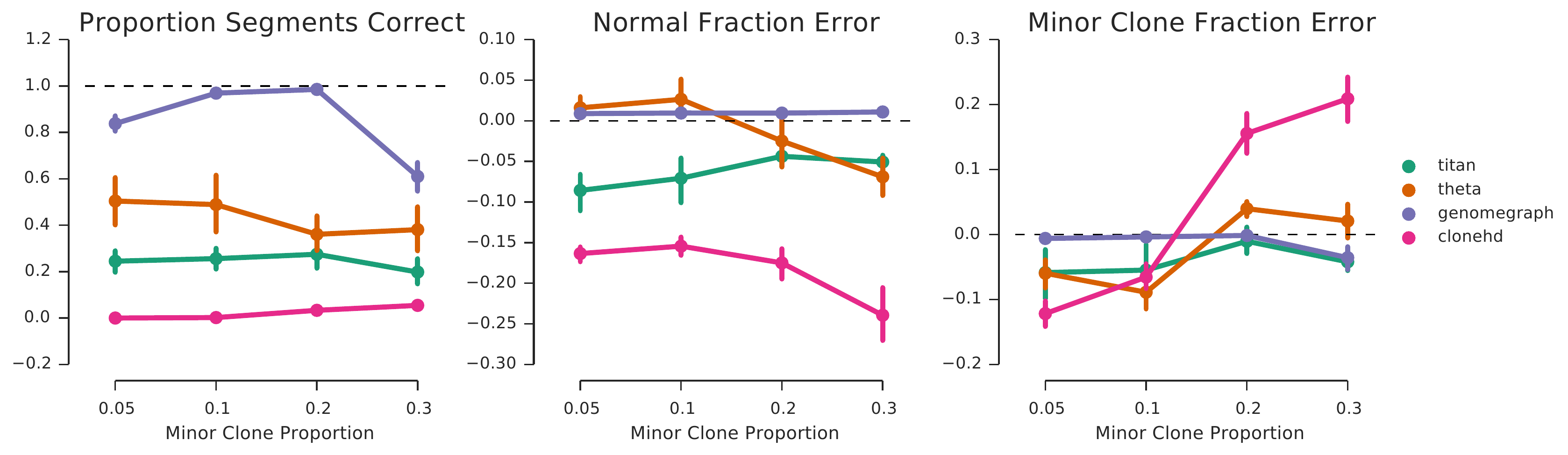}
		\label{fig:inferminorclone}
	\end{subfigure}
	\begin{subfigure}[b]{\textwidth}
		\centering
		\caption{Proportion of genomic segments with divergent (subclonal) copy number}
		\includegraphics[width=0.75\textwidth]{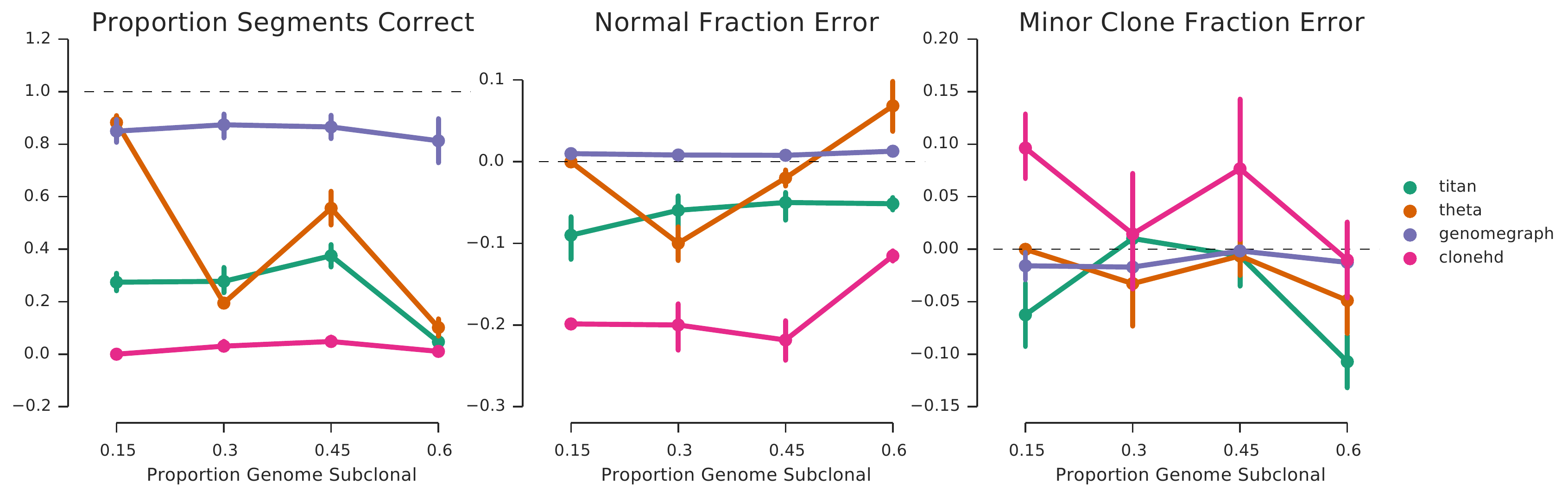}
		\label{fig:inferpropsubclonal}
	\end{subfigure}
\end{figure}

We considered the possibility that other tools may perform poorly on datasets with significant amounts of high level amplification.
Thus we simulated an additional 40 genome mixtures, modifying the target copy number state distribution to produce genomes for which segments with more than 4 total copies are very rare.
For this simulation we fixed the proportion of the genome that is divergent at 0.25, and varied the minor clone proportion in the mixture.
As can be seen in Figure~\ref{fig:comparisonresults2}, each tool performs better in some circumstances, but only the genomegraph method performs consistently well across different mixtures.

\begin{figure}[h]
	\caption{Performance comparison of the genomegraph algorithm with TITAN, CloneHD, and Theta2.0, using a dataset with limited amplified regions.  Showed is performance when varying the proportion of minor clone in the mixture with normal fixed at 0.4}
	\centering
	\includegraphics[width=0.75\textwidth]{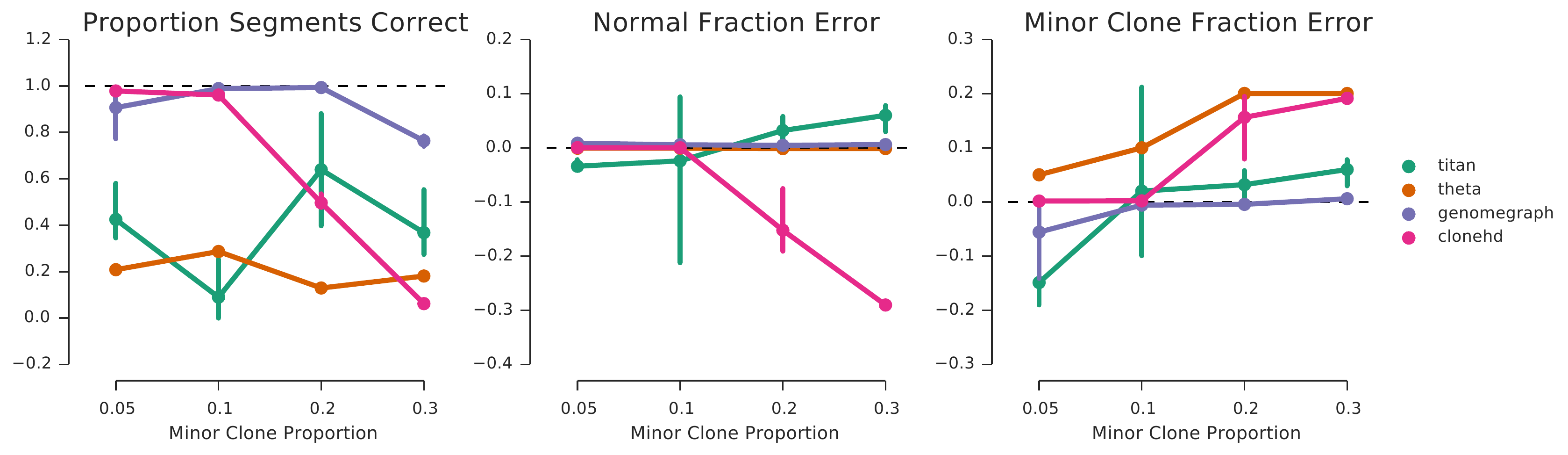}
	\label{fig:comparisonresults2}
\end{figure}

\section{Discussion}

We have developed a novel method for joint inference of genome content and structure.
Using a comprehensive set of simulated genome mixtures, we have shown that joint inference out-performs naive methods with respect to identification of subclonal breakpoints and classification of breakpoints as real or false positive.
Furthermore, we have shown that inclusion of breakpoints during copy number inference provides a modest but consistent improvement in the accuracy of predicted segment copy number.
Results from a comparison against existing subclonal copy number inference tools shows our method out-performs existing methods on the simulated mixtures.

Our method provides several additional novel contributions.
When selecting a segment length, a balance must be struck between the need for increased statistical strength provided by longer segments, and the additional noise that results from a true copy number transition occurring in the middle of a segment.
By segmenting at the break-ends of rearrangement breakpoints, we attempt to capture a majority of the copy number transitions with our segmentation.
We then use a medium size segmentation to break up longer unsegmented regions, allowing us to retain improved statistical strength, while modeling for a majority of the changepoints.
Furthermore, haplotype blocks are used in a novel way to improve the accuracy of allele specific read counts.
The aggregation of allele specific read counts across segments allows us to model the alleles of large segments as (approximately) independent, making inference more tractable.
Compared to copy number inference tools such as Titan, our state space is in some ways more comprehensive, allowing for any combination of clone copy number at each segment.

Though our initial results are promising, significant work remains.
The proposed moves of our greedy heuristic are not sufficient to escape local optima of significant importance.
For instance, a breakage fusion bridge cycle would be represented as a loop bond edge in the genome graph.
We do not include loop bond edges as they would never be selected by a matching based approach.
To support loops, we could either add dummy segment edges, making each loop a 3 edge alternating path, or use a more general matching algorithm such as minimum perfect b-matchings to identify optimal moves.
Finally, future work may show benefit to inclusion of rearrangement breakpoints during learning, providing motivation for development of a method that jointly infers genome content, genome structure, and clone structure.

\section{Supplementary Methods}

\subsection{Expectation Maximization Method for Learning $h$}
\label{sec:learnh}

To infer the haploid read depth parameter $h$, we break the dependency between segments connected by breakpoints, and rely on a restricted model that includes only the dependencies between segments adjacent in the reference genome.
We use a similar likelihood function, and then select a tractable graphical model (HMM) structure that matches the original problem as closely as possible.
Hidden states in the HMM correspond to copy number matrices: for segment $n$, copy number state $c$ is an $M$ by $2$ matrix, with entry $c_{m \ell}$ for the copy number of allele $\ell$ for clone $m$.

The HMM version of our problem requires a finite state space.
We restrict the state space by imposing a maximum copy number, and a maximum copy number difference between clones.
Segments with true copy number outside this finite state space are modeled with an additional null state $\emptyset$.
Transitions into the null state are penalized.
Segments in the null state are free to take on any copy number state $b_n$ in the infinite state space of all copy number matrices.
Thus $x_n$ is dependent on both the copy number prior $C$ and an independent copy number state $b_n$.
A uniform (improper) prior is placed over $b_n$.
In practice, only $b_n$ in the neighborhood of $\operatornamewithlimits{argmax}_b p(x_n|h, c_n=b)$ are considered.
The graphical model for the HMM is depicted in Figure~\ref{fig:hmmgraphicalmodel}.
\begin{figure}[h]
	\centering
	\includegraphics{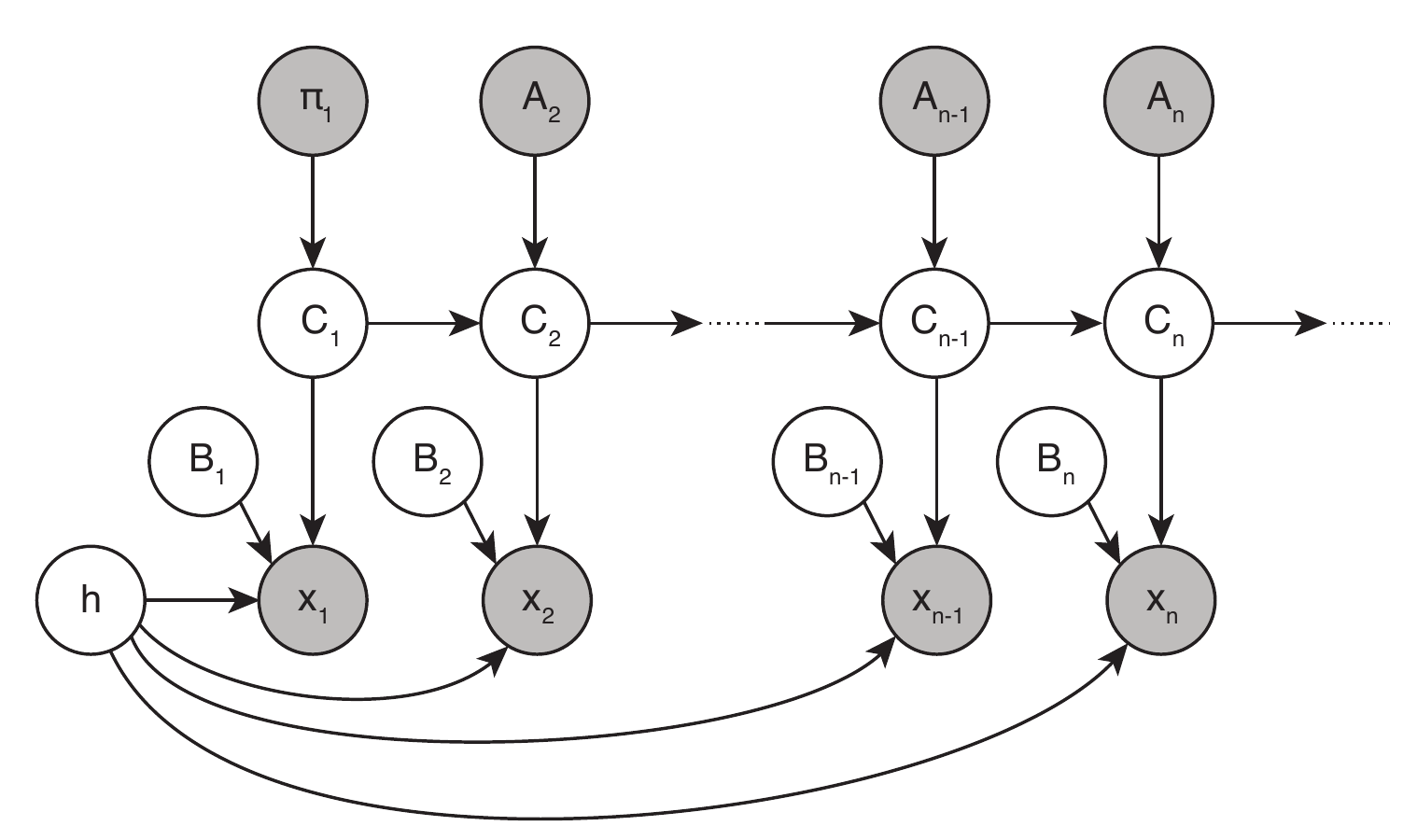}
	\caption[graphical model]
	{Graphical model of the HMM version of the problem used for parameter inference.}
	\label{fig:hmmgraphicalmodel}
\end{figure}

The full data log likelihood can be expressed as given by Equation~\ref{eqn:fulldataloglikelihood}, where $q(b,c)$ is given by Equation~\ref{eqn:copynumberchoice}.
\begin{eqnarray}
\log p(X, B, C | h, \cdot) &=& \log p(X | B, C, h, \cdot) + \log p(B | \cdot)  + \log p(C | \cdot) \nonumber \\
&=& \sum_n \sum_b \sum_c \mathbb{I}(b_n=b, c_n=c) \log p(x_n | h, c_n=q(b,c), \cdot) \nonumber \\
&& + \sum_n \sum_b \mathbb{I}(b_n=b) \log p(b_n=b) \nonumber \\
&& + \sum_c \mathbb{I}(c_1=c) \log \pi_1(c) \nonumber \\
&& + \sum_{n=2}^{N} \sum_c \sum_{c'} \mathbb{I}(c_{n-1}=c, c_n=c') \log A_n (c, c') 
\label{eqn:fulldataloglikelihood}
\end{eqnarray}
\begin{eqnarray}
q(b,c) &=& 
\begin{cases}
c &\mbox{if } c \neq \emptyset \\
b &\mbox{if } c = \emptyset
\end{cases} \label{eqn:copynumberchoice}
\end{eqnarray}

We model the probability of a segment having copy number that is higher than the maximum copy number in our HMM state space with an exponential distribution over the length of the segment, with exponential parameter $\lambda_a$.
Calculate $\alpha_n(c)$, the probability of high level amplification of segment $n$, as given by Equation~\ref{eqn:ampprior}.
\begin{eqnarray}
\alpha_n(c) &=& \lambda_a \exp \left( - \lambda_a l_n \mathbb{I}(c=\emptyset) \right) 
\label{eqn:ampprior}
\end{eqnarray}

Additionally, we model the probability of a segment having copy number that is divergent between tumour clones with an exponential distribution over the length of the segment, with exponential parameter $\lambda_d$.
Calculate the number of alleles with divergent copy number $d(c)$ as given by equation \ref{eqn:divergentcount}, and the probability of divergence $\rho_n(c)$ as given by equation \ref{eqn:divergentprior}.
\begin{eqnarray}
d(c) &=& \sum_{\ell} \prod_{m=2}^M \prod_{m'=2}^M \mathbb{I}(c_{m \ell} =c_{m' \ell}) 
\label{eqn:divergentcount} \\
\rho_n(c) &=& \lambda_d \exp \left( - \lambda_d l_n d(c) \right) 
\label{eqn:divergentprior}
\end{eqnarray}

Finally, we model the probability of copy number transitions between adjacent segments as an exponential distribution over the copy number difference between the segments.
The parameter of the transition exponential, $\beta$ is set to the same value as the telomere penalty in the combinatorial problem.
Calculate the number of telomere copies $t(c, c')$ at a transition from a segment with copy number $c'$ to a segment with copy number $c$ as given by equation \ref{eqn:telomerecopynumber}, and the probability of the transition as given by equation \ref{eqn:telomeretransitionprob}.
\begin{eqnarray}
t(c, c') &=& \sum_m \sum_{\ell} |c-c'|
\label{eqn:telomerecopynumber} \\
\eta(c, c') &=& \beta \exp\left( - \beta t(c, c') \right)
\label{eqn:telomeretransitionprob}
\end{eqnarray}

The prior probability over the copy number of the first segment $p(c_1=c)$ is thus given by equation \ref{eqn:hmmprior}, and the transition probability from segment $n-1$ to segment $n$ is given by equation \ref{eqn:hmmtransition}.
\begin{eqnarray}
\pi_1 (c) &=& p(c_1=c) \nonumber \\
&\propto& \rho_1(c) \alpha_1(c)
\label{eqn:hmmprior} \\
A_n (c, c') &=& p(c_n=c | c_{n-1}=c') \nonumber \\
&\propto& \rho_n(c) \alpha_n(c) \eta(c, c')
\label{eqn:hmmtransition}
\end{eqnarray}

The expected value of the log likelihood function with respect to the conditional distribution $p(C|X,h^{(t-1)},\cdot)$, excluding terms independent of $h$, is given by Equation~\ref{eqn:expectationfunction}.
\begin{eqnarray}
Q(h, h^{(t-1)}) &=& \mathbb{E}_{p(B,C|X,h^{(t-1)},\cdot)} [\log p(X,B,C|h,\cdot)] \nonumber \\
&=& \sum_n \sum_b \sum_c p(b_n=b,c_n=c|x_n,h^{(t-1)},\cdot) \log p(x_n | h, c_n=q(b,c), \cdot) \nonumber \\
&=& \sum_n \sum_{c \neq \emptyset} p(c_n=c|x_n,h^{(t-1)},\cdot) \log p(x_n | h, c_n=c, \cdot) \nonumber \\
&& + \sum_n \sum_b p(b_n=b,c_n=\emptyset |x_n,h^{(t-1)},\cdot) \log p(x_n | h, c_n=b, \cdot)
\label{eqn:expectationfunction}
\end{eqnarray}

To calculate the joint-posterior marginal probabilities $p(b_n=b,c_n=c|x_n,h^{(t-1)},\cdot)$, and thereby calculate $Q(h, h^{(t-1)})$ (E-step), we use the forwards-backwards algorithm.
The graphical model is a polytree, thus the forwards-backwards cannot be applied directly.
Instead, we model both $c_n$ and $b_n$ jointly, and apply the forwards-backwards to the merged state space of $c_n$ and $b_n$.

Typically the state space of the jointly modeled variables would be a cross product of the state space of each variable independently.
For our model, such a state space would be unmanageably large.
Fortunately a cross product state space can be avoided due to the mutually exclusive effect of $c_n$ and $b_n$ on the likelihood of $x_n$, and given that $c_n$ and $b_n$ are independent of $b_{n-1}$.
For instance, when marginalizing over the space of $c_n$ and $b_n$ in the forwards pass, part of the marginalization can take place analytically (Equation~\ref{eqn:forwarditer}).
\begin{eqnarray}
p(b_n=b, c_n=c | x_{1:n-1}) &=& \sum_{b'} \sum_{c'} p(b_{n-1}=b', c_{n-1}=c' | x_{1:n-1}) \nonumber \\
&& \times \ p(b_n=b, c_n=c | b_{n-1}=b', c_{n-1}=c') \nonumber \\
&=& \sum_{c' \neq \emptyset} p(c_{n-1}=c' | x_{1:n-1}) p(b_n=b, c_n=c | c_{n-1}=c') \nonumber \\
&& + \sum_{b'} p(b_{n-1}=b', c_{n-1}=\emptyset | x_{1:n-1}) p(b_n=b, c_n=c | c_{n-1}=\emptyset) \nonumber \\
\label{eqn:forwarditer}
\end{eqnarray}

The backwards pass can be simplified similarly.
In effect, we can consider a state space that includes all regular ($\neq \emptyset$) states for $c_n$, concatenated with the additional states for $b_n$.
Entries of the transition matrix are given by Equation~\ref{eqn:modtransition}.
\begin{eqnarray}
p(b_n=b, c_n=c | c_{n-1}=c') &=& p(c_n=c | c_{n-1}=c') p(b_n=b)
\label{eqn:modtransition}
\end{eqnarray}

We maximize $Q(h, h^{(t-1)})$ with respect to $h$ (M-Step) numerically.
The gradient of $Q(h, h^{(t-1)})$, used for numerical maximization, can be computed as given by Equations~\ref{eqn:maximizeq1}-\ref{eqn:maximizeq4}.
\begin{eqnarray}
\frac{\partial Q(h, h^{(t-1)})}{\partial h_m} &=& \sum_n \sum_{c \neq \emptyset} p(c_n=c|x_n,h^{(t-1)},\cdot) \sum_{k=1}^3 \frac{\partial}{\partial h_m} \log p(x_{nk} | h, c_n=c, \cdot) \nonumber \\
&& + \sum_n \sum_b p(b_n=b,c_n=\emptyset|x_n,h^{(t-1)},\cdot) \sum_{k=1}^3 \frac{\partial}{\partial h_m} \log p(x_{nk} | h, c_n=b, \cdot) \nonumber \\
\label{eqn:maximizeq1} 
\end{eqnarray}
\begin{eqnarray}
\frac{\partial}{\partial h_m} \log p(x_{nk} | h, c_n, \cdot) &=& \left( \frac{\partial}{\partial \mu_{nk}} \log p(x_{nk} | h, c_n, \cdot) \right) \left( \frac{\partial \mu_{nk}}{\partial h_m} \right)
\label{eqn:maximizeq2} \\
\frac{\partial \mu_{nk}}{\partial h_m} &=& \frac{\partial \gamma_n}{\partial h_m} P_{n} l_n
\label{eqn:maximizeq3} \\
\frac{\partial \gamma_{nk}}{\partial h_m} &=& c_{nm}
\label{eqn:maximizeq4}
\end{eqnarray}
 
A negative binomial specific term is given by Equation~\ref{eqn:derivativenegbin}, and the poisson specific term by Equation~\ref{eqn:derivativepoisson}.
\begin{eqnarray}
\frac{\partial}{\partial \mu_{nk}} \log p(x_{nk} | h, c_n, \cdot) &=& \frac{x_{nk}}{\mu_{nk}} - \frac{r_k + x_{nk}}{r_k+\mu_{nk}}
\label{eqn:derivativenegbin} \\
\frac{\partial}{\partial \mu_{nk}} \log p(x_{nk} | h, c_n, \cdot) &=& \frac{x_{nk}}{\mu_{nk}} - 1 
\label{eqn:derivativepoisson}
\end{eqnarray}

\subsection{Estimating the overdispersion parameter $r$}
\label{sec:overdispersion}

We estimate the overdispersion parameter $r$ offline from segment read counts.
We assume that the majority of adjacent segments have the same genotype, and thus the same expected read depth $\gamma$.
Under this assumption, we identify the $r$ that maximize the likelihood of the read count data (Equation~\ref{eqn:inferrlikelihood}) for pairs of adjacent segments $i$ and $i+1$ with identical read depth $\gamma_i$.
We use gradient descent to find a local optima of the likelihood with respect to both $r$ and $\gamma_i$, with partial derivates with respect to $r$ and $\gamma_i$ calculated as given by Equations~\ref{eqn:inferrpartialr} and \ref{eqn:inferrpartialgamma} respectively.

\begin{eqnarray}
\ell(r,l,p,\gamma) &=& \sum_{i=1}^{\frac{N}{2}} \sum_{n=2i-1}^{2i} \log p(x_n | l_n, p_n, \gamma_i, r) \nonumber \\
&=&  \sum_{i=1}^{\frac{N}{2}} \sum_{n=2i-1}^{2i} \log{(\Gamma(x_n + r))} - \log(x_n !) - \log{(\Gamma(r))} + r \log{(r)} \nonumber \\ && - r \log{(r+l_n \gamma_i)} + x_n \log({l_n \gamma_i}) - x_n \log{(r+l_n \gamma_i)} 
\label{eqn:inferrlikelihood} \\
\frac{\partial \ell(r,l,p,\gamma)}{\partial \gamma_i} &=& \sum_{n=2i-1}^{2i} -\frac{r l_n}{r + l_n \gamma_i} + \frac{x_n}{\gamma_i} - \frac{x_n l_n}{r + l_n \gamma_i}
\label{eqn:inferrpartialgamma}\\
\frac{\partial \ell(r,l,p,\gamma)}{\partial r} &=& \sum_{i=1}^{\frac{N}{2}} \sum_{n=2i-1}^{2i} \psi(x_n + r) - \psi{(r)} + \log{(r)} + 1 \nonumber \\ && - \log{(r + l_n \gamma_i)} - \frac{r}{r + l_n \gamma_i} - \frac{x_n}{r + l_n \gamma_i} 
\label{eqn:inferrpartialr}
\end{eqnarray}

\subsection{Independence of Segment Read Counts}
\label{sec:segmentindependence}

Previously, \cite{Oesper:2013kq} modelled segment read counts as a single draw from a multinomial likelihood.
Suppose, in addition to the multinomial, we model the total number of reads $T$ measured by the sequencing experiment as a Poisson distributed random variable with unknown mean $\lambda$.
The joint likelihood of multinomial distributed read counts $x_j$ and total read count $T$ given proportions $\pi_i$ and expected total read count $\lambda$ can be written as given by Equation~\ref{eqn:multinomial}.
Introduce variables $\lambda_j$ such that $\lambda = \sum_j \lambda_j$ and $\pi_j = \frac{\lambda_j}{\lambda}$.
Then Equation~\ref{eqn:multinomial} can be rewritten as given by Equation~\ref{eqn:multinomial_poisson}, which is the likelihood of $J$ poisson distributed independent random variables with means $\lambda_j$ \cite{forster2010bayesian}.
Thus, our use of independent Poisson likelihoods can be seen as equivalent to the multinomial used by \cite{Oesper:2013kq} if we also assume the total read count of the experiment is a Poisson distributed random variable.
\begin{eqnarray}
p(T, x_1,..,x_J| \lambda, \pi_1,..,\pi_J) &\propto& e^{-\lambda} \lambda^T \prod_j \pi_j^{x_j} \label{eqn:multinomial} \\
p(x_1,...,x_J| \lambda_1,...,\lambda_J) &\propto& e^{-\lambda} \prod_j \lambda_j^{x_j} \label{eqn:multinomial_poisson}
\end{eqnarray}

\subsection{Parameters used for existing methods}
\label{sec:runtools}

\subsubsection{Theta2.0}

Theta requires an existing segmentation, thus we merged adjacent segments with identical copy number states to form a collection of perfect segments.
We then counted reads within those segments and used those read counts as input to Theta.
We used the \texttt{--NUM\_INTERVALS 15} option to allow for reasonable running time, and the \texttt{--FORCE} argument to for Theta to run on some of the genomes which were sub-optimal candidates for Theta analysis.
We then used octave to execute the \texttt{runBAFGaussianModel} function to select a single solution from potentially multiple optimal solutions.
Only the solution with 2 tumour clones was considered, as that is what was simulated.

\subsubsection{Titan}

As input to Titan, we calculated counts of reads contained within regular 1000bp segments, and counts reads supporting the reference and non-reference allele for heterozygous germline SNPs.
We then ran Titan, without correcting for GC and mappability (as this was not simulated).
We used multiple initializations, with normal contamination from 0 to 1 in increments of 0.1, and ploidy from 1 to 4 in increments of 1.
The number of Titan clusters was fixed at 2.
We then selected the solution with lowest \texttt{S\_Dbw validity index}.
Since the parameterization for Titan is slightly different than for other tools, we used the following formula to convert from estimates of tumour clone prevalences $t_1$ and $t_2$, and normal contamination estimate $n$, as follows.

\begin{eqnarray}
	\operatorname{mixture} = [n, (1-n) \times t_2, (1-n) \times |t_1 - t_2|]
\end{eqnarray}

\subsection{CloneHD}

As input to CloneHD, we calculated counts of reads contained within regular 1000bp segments, and counts reads supporting the reference and non-reference allele for heterozygous germline SNPs (b-allele frequency (BAF) data).
We used a series of steps as outlined in \texttt{run\_example.sh}.
\begin{enumerate}
	\item use filterhd to analyze the normal read depth data for technical read depth modulation
	\item use filterhd to analyze the tumour read depth data to get a benchmark of the log likelihood
	\item use filterhd to analyze the tumour read depth data with the bias estimate from the normal
	\item use filterhd to analyze the tumour BAF data
	\item use clonehd to infer copy number based on tumour read depth and BAF data.
\end{enumerate}

\bibliographystyle{plain}
\bibliography{demix}

\end{document}